\documentclass{ifacconf}

\usepackage{graphicx}      
\usepackage{natbib}        
\usepackage{amssymb} 
\usepackage{amsmath} 
\usepackage{algorithm}
\usepackage[ruled,vlined,algo2e]{algorithm2e}
\usepackage{soul}
\usepackage[noend]{algpseudocode}
\usepackage{epsfig} 
\usepackage{textcomp}
\usepackage{lipsum}
\usepackage{arydshln}
\usepackage{etoolbox}
\usepackage{ mathrsfs }
\usepackage{gensymb}
\usepackage{color}
\usepackage{mathtools}
\usepackage{systeme}
\usepackage{amsfonts}
\usepackage{empheq}
\usepackage{float}
\usepackage{caption}
\usepackage{subcaption}
\usepackage{bbold}
\usepackage[outdir=./]{epstopdf}
\usepackage{enumitem}

\newtheorem{definition}{Definition}
\newtheorem{lemma}{Lemma}
\newtheorem{assumption}{Assumption}
\newtheorem{theorem}{Theorem}{}
\newtheorem{proof}{Proof}{}
{}
\newtheorem{remark}{Remark}{}
\newtheorem{problem}{Problem}{}
\begin{document}
\begin{frontmatter}

\title{Safety monitoring under stealthy sensor injection attacks using reachable sets} 

\author[First]{Cédric Escudero} 
\author[Second]{Michelle S. Chong} 
\author[First]{Paolo Massioni}
\author[First]{Eric Zamaï}

\address[First]{Univ Lyon, INSA Lyon, Université Claude Bernard Lyon 1, Ecole Centrale de Lyon, CNRS, Ampère, UMR5005, 69621 Villeurbanne, France (e-mail: cedric.escudero@insa-lyon.fr, paolo.massioni@insa-lyon.fr, eric.zamai@insa-lyon.fr)}
\address[Second]{Eindhoven University of Technology, (e-mail: m.s.t.chong@tue.nl)}
   
\begin{abstract}                
Stealthy sensor injection attacks are serious threats for industrial plants as they can compromise the plant's integrity without being detected by traditional fault detectors. In this manuscript, we study the possibility of revealing the presence of such attacks by monitoring only the control input. This approach consists in computing an ellipsoidal bound of the input reachable set. When the control input does not belong to this set, this means that a stealthy sensor injection attack is driving the plant to critical states. The problem of finding this ellipsoidal bound is posed as a convex optimization problem (convex cost with Linear Matrix Inequalities constraints). Our monitoring approach is tested in simulation.
\end{abstract}


\end{frontmatter}

\section{Introduction}
Information and communication technologies are progressively integrated in many industrial applications. In particular, the control of plants is targeted by this integration to improve the control performance and the awareness of their surrounding environment in making appropriate decisions. As a result, plants are now controlled by embedded computer devices over communication networks, commonly called Network Control Systems (NCSs) (\cite{Lee_CPS}). Beyond the benefits made possible by information and communication technologies, new issues have emerged. A critical one is recently highlighted by academics as well as industrial stakeholders: the security issue (\cite{Cardenas_2009}). NCSs are exposed to vulnerabilities (e.g. for remote control or maintenance of the embedded computer system via the internet or cellular networks), and can be exploited by adversaries to launch cyber attacks (\cite{Firoozjaei2022}). In particular, \textit{deception} attacks have attracted the attention of the control engineering community. These attacks aim at degrading the control performance by tampering with system's signals (sensing and control). This includes spoofing, false-data injection, and replay attacks (\cite{Teixeira}) that can have severe consequences on the plant.

The secure control literature (\cite{Chong_2019, annual_rev_Sandberg_2021}) has rapidly grown to come up with different methods to detect deception attacks. Some works have proposed quantifying the impact of deception attacks on the plant by means of set-theoretic methods (\cite{Milosevic_2020, Liu_2021, Zhang_2022}). Set-theoretic methods involve computing sets to encapsulate the reachable states of dynamical systems; this is the so-called reachability analysis. Other works, based on  reachability analysis, have proposed methods to mitigate the impact of stealthy attacks. Stealthy attacks are attacks that aim to damage the plant while avoiding dectection (e.g. fault detectors). \cite{MurguiaAutomatica} and \cite{Li_2022b} propose synthesis methods for the control block (i.e. the controller and the fault detector) to minimize the impact of stealthy attacks. Besides redesigning the control block, some works design monitoring approaches to reveal the presence of attacks. Among those works, \cite{Azzam_2022} proposes a predictive approach to check the safety of the plant in the presence of potential stealthy attacks. \cite{Yang_2021} uses an observer-based estimator using a bank of unknown input observers to estimate the plant states and detect attacks. \cite{Hu_2019} reveals the presence of stealthy attacks by analyzing the residual distribution.

In this manuscript, we focus on \textit{stealthy sensor injection attacks}, a class of deception attacks that injects malicious sensing signals into the network while escaping detection by the fault detector. We assume an adversary is capable of manipulating sensor measurements by compromising the communication network between the plant and the controller. The closest work relative to such attacks is in \cite{MurguiaAutomatica} where the authors model the attack and synthesize the control block to mitigate the impact of attacks on the plant. Here instead, we want to reveal the presence of stealthy sensor injection attacks by monitoring the control signals only. Indeed, sensor injection attacks will, at some time, affect the control signals to drive the plant into critical states, which will damage the plant integrity. Moreover, compared to estimator-based approaches, our approach will only require one signal to monitor (i.e., the control signal), hence reducing the vulnerabilities of the monitoring approach. To do so, we propose a set-theoretic method based on reachable sets. It consists first in finding the state reachable set when the system is subject to stealthy sensor attacks and the plant state are constrained by a given safe set. The safe set is the set of plant states where its safe and proper operation is guaranteed. The resulting state reachable set encompasses all the trajectories of the plant states under stealthy attacks that are safe. Then, we search for an input reachable set that guarantees the plant states belong to the computed state reachable set for all stealthy attacks. Hence, the monitoring approach consists in verifying whether the control signal does not belong to this input reachable set, which we can then verify that the control system is under stealthy attacks due to our computationally efficient outer-approximation of the input reachable set. The computation of tight bounds on reachable sets is challenging in general (see \cite{kurzhanski2000ellipsoidal}), especially when online computation is paramount, as is the case for critical control systems.

The rest of this manuscript is organized as follows. Section~\ref{sec:problem} formulates the research problem we want to address. Section~\ref{sec:tools} provides the necessary background about computing the ellipsoidal bound of reachable sets. Our approach of monitoring the input signals is presented in Section~\ref{sec:result_sec}. Lastly in Section~\ref{sec:example}, we apply our results on a three-tank system to illustrate the performance of our monitoring approach.


\noindent
\emph{\textbf{Notation:}} The symbol $\mathbb{R}$ stands for the real numbers, $\mathbb{R}^{n \times m}$ is the set of real ${n \times m}$ matrices, and $\mathbb{R}_{>0}$ ($\mathbb{R}_{\geq 0}$) denotes the set of positive (non-negative) real numbers. Matrix $A^\top$ indicates the transpose of matrix $A$ and diag($a_1,...,a_n$) corresponds to a diagonal matrix with diagonal elements $a_1,...,a_n$.  The identity matrix of dimension $n$ is denoted by $I_n$, and $\mathbf{0}$ is a matrix of only zeros of appropriate dimensions. The notation $A \succeq 0$ (resp. $A \preceq 0$) indicates that the matrix $A$ is positive (resp. negative) semidefinite, i.e., all the eigenvalues of the symmetric matrix $A$ are positive (resp. negative) or equal to zero, whereas the notation \(A \succ 0\) (resp. \(A \prec 0\)) indicates the positive (resp. negative) definiteness, i.e., all the eigenvalues are strictly positive (resp. negative). The notation $\mathcal{E}_\varphi(\Phi,\bar{\phi})$ stands for an ellipsoidal set of dimension $n_\phi$ with shape matrix $\Phi \in \mathbb{R}^{n_\varphi \times n_\varphi},\, \Phi = \Phi^\top  \succ 0$ and centered at $\bar{\phi}$, i.e.,
 $\mathcal{E}_\varphi(\Phi, \bar{\phi}):=\{\varphi \in \mathbb{R}^n\, |\,(\varphi - \bar{\phi})^\top \Phi (\varphi - \bar{\phi})  \leqslant 1  \}$; if no center $\bar{\phi}$ is specified in the ellipsoid notation this means that the ellipsoid is centred at $0$, i.e.  $\mathcal{E}_\varphi(\Phi):=\{\varphi \in \mathbb{R}^n\, |\,\varphi^\top \Phi \varphi  \leqslant 1  \}$.
 
\section{Problem formulation} \label{sec:problem}
In this section, we introduce the class of systems and attacks under study. Then, we introduce at a high-level our proposed solution to detect those attacks based on input monitoring, and formulate the research problem addressed in this manuscript.

\subsection{System dynamics}
We consider a system  $\Sigma_p$ with the following dynamics
\begin{align} \label{eq:plant}
    \dot{x}_p(t) & = A_p x_p(t) + B_p u(t) \nonumber \\
    y(t) & = C_p x_p(t) + D_p u(t),
\end{align}
with time $t \in \mathbb{R}_{>0}$, plant state $x_p \in \mathbb{R}^{n_p}$, sensor measurement $y\in\mathbb{R}^{l}$, control input $u \in \mathbb{R}^{m}$, and system matrices $A_p, \, B_p, \, C_p, \, D_p$ with appropriate dimensions.

The system description in \eqref{eq:plant} representing the plant is part of a closed-loop as illustrated in Figure~\ref{fig:setup}. The system receives the control action $u$ and transmits the sensor measurement $y$ through a public/unsecured communication network.

\subsection{Adversarial capabilities}
In this manuscript, we consider an adversary who can inject sensor data from the unsecured network. That is, it can add some signals $\delta y(t)$ to the true sensor measurement $y(t)$, i.e. $\delta y(t)$ models the corruption of sensor measurements. The adversary can compromise up to $s$ sensor measurements $y(t)$, with $s \in \{1, \ldots, l\}$. We also denote by $\Gamma$, an attacker's selection matrix to choose how the additive signals $\delta y(t)$ affects the true sensor measurements as proposed in \cite{MurguiaAutomatica}. The received sensor measurement, i.e., the true sensor measurement with the additional signals controlled by the adversary, denoted $\tilde{y}(t)$ takes the following form:
\begin{align} \label{eq:attack}
    \tilde{y}(t) & = y(t) + \Gamma \delta y(t),
\end{align}
where $\delta y(t) \in \mathbb{R}^s$ denotes additive sensor measurement.

Notice that the attacker's selection matrix $\Gamma$ depends on the attacker and defender capabilities. If the defender secures some output channels, for instance with encryption, any data injection attack will not affect the received decrypted data. Similarly, the attacker wants to reduce its attack cost, then it might select only a few channels to perform its attack. The worst-case scenario for the defender is when $\Gamma = I_l$, which is when the attacker attacks all sensors.

\begin{figure}[t!]
\centering
\includegraphics[width=0.3\textwidth]{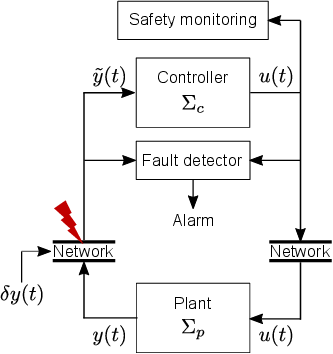}
\caption{Control system in the presence of sensor corruption.}
\label{fig:setup}
\end{figure}

\subsection{Dynamic output feedback controller}
A controller $\Sigma_c$ receives the sensor measurement $\tilde{y}$ through the unsecured communication network and computes control actions $u$ which are transmitted to the plant.

The control input to the plant $u:\mathbb{R}_{\geq 0} \to \mathbb{R}^{m}$, affected by the sensor attacks, is given by a controller $\Sigma_c$ with the following dynamics
\begin{align} \label{eq:control}
    \dot{x}_{c}(t) & = A_c x_c(t) + B_c \tilde{y}(t) \nonumber  \\
    u(t) & = C_c x_c(t) + D_c \tilde{y}(t),
\end{align}
with controller state $x_c\in\mathbb{R}^{n_c}$.

We state all the assumptions about the plant and controller models below.

\begin{assumption}[Modelling assumptions] \label{assum:model}
The plant \eqref{eq:plant} and controller \eqref{eq:control} models satisfy the following:
\begin{enumerate}[label=(\roman*)]
    \item The matrix $I_{l \times l}-D_p D_c$ is invertible such that the closed-loop system \eqref{eq:plant} and \eqref{eq:control} is well-posed.
    \item The plant \eqref{eq:plant} interconnected with the controller \eqref{eq:control} is asymptotically stable. \hfill $\Box$
\end{enumerate}
     
\end{assumption}
When $D_p = \mathbf{0}$, condition (i) in Assumption \ref{assum:model} holds.

\subsection{Fault detector}
Traditionally, control systems are equipped with fault detectors to detect the occurrence of faults in the plant. Most model-based fault detectors consist in computing the residual, that is the difference between the sensor measurement and the estimated measurement. Here, to estimate the plant state $x_p$, we use the following state estimator:
\begin{align} \label{eq:filter}
    \dot{\hat{x}}(t) & = A_p \hat{x}(t) + B_p u(t) + L(\tilde{y}(t) - \hat{y}(t)) \nonumber \\
    \hat{y}(t) & = C_p \hat{x}(t) + D_p u(t),
\end{align}
with the estimated plant state $\hat{x}\in\mathbb{R}^{n_{p}}$, and gain matrix $L \in \mathbb{R}^{n_{p} \times l}$.

\noindent Let $e(t)$ denote the state estimation error defined as $e(t) := x_p(t) - \hat{x}(t)$, and $r(t)$ denote the residual defined as $r(t):=\tilde{y}(t)-\hat{y}(t)$. From the plant and controller models in \eqref{eq:plant}, \eqref{eq:control}, the detector's estimation error and the residual evolve as follows:
\begin{align} \label{eq:fault_detector}
    \dot{e}(t) & = (A_p-L C_p) e(t) - L \Gamma \delta y(t) \nonumber  \\
    r(t) & = C_p e(t) + \Gamma \delta y(t),
\end{align}

\begin{assumption}[Modelling assumptions] \label{assum:fd}
The fault detector \eqref{eq:fault_detector} model satisfies the following:
\begin{enumerate}[label=(\roman*)]
    \item The pair ($A_p$, $C_p$) is observable such that there exists an $L$ of appropriate dimensions such that ($A_p-L C_p$) is Hurwitz.  \hfill $\Box$
\end{enumerate}   
\end{assumption}
Therefore, under Assumption \ref{assum:fd}, when there are no sensor attacks ($\delta_y(t)=0$, for all $t\geq 0$), the detector's estimation error $e$ and residual $r$ will converge to the origin. Thereby, indicating that there is indeed no fault.

In fault detection, the objective is to decide between the two following hypotheses:
\begin{enumerate}
\item $\mathcal{H}_0$: the system is normal (without faults/attacks)
\item $\mathcal{H}_1$: the system is faulty
\end{enumerate}

We consider here a fault detector that raises an alarm if $r(t)^\top \Pi r(t) > 1$ at some time $t$, with $\Pi \succ 0$. This means that we are under the normal mode when $r(t)^\top \Pi r(t) \leq 1, \, \forall \, t \in\mathbb{R}_{>0}$.

\subsection{Stealthiness and attack definition}
In this manuscript, we consider an adversary who wants to remain stealthy with respect to the fault detector in \eqref{eq:fault_detector}. That is, it wants to launch \textit{stealthy} sensor injection attacks that aim to avoid raising an alarm in the fault detector, i.e., being in the normal mode $\mathcal{H}_0$. Stealthy attacks are constrained attacks in the sense that the adversary must carefully choose the attack signal $\delta y(t)$ such that $r(t)^\top \Pi r(t) \le 1,\, \forall t \ge 0$. Stealthy attacks often cannot lead to quick in time and high impact damage to the plant, but are long term attacks (e.g. causing the rapid aging of equipment). We define the stealthy set $\mathcal{E}_r(\Pi)$ as
\begin{align} \label{eq:stealthy}
r(t)^\top \Pi r(t) \leq 1 \Leftrightarrow r(t) \in \mathcal{E}_r(\Pi)
\end{align}

If \eqref{eq:stealthy} is satisfied for any attack signal $\delta y(t)$, then no alarms will be raised in the fault detector, that is sensor injection attacks are stealthy.

From the residual dynamics in \eqref{eq:fault_detector}, it yields that $\delta y(t) = \Gamma^{+}(r(t) - C_p e(t))$, where $\Gamma^{+}$ denotes the Moore-Penrose inverse of $\Gamma$ (see \cite{MurguiaAutomatica} for more details). We can now state our closed-loop system from plant and controller dynamics in \eqref{eq:plant}, \eqref{eq:control} together with state estimation error of the fault detector in \eqref{eq:fault_detector} by taking the extended state $\zeta:=[x_p^\top, x_c^\top, e^\top]^\top$, with $\zeta \in \mathbb{R}^n$ where $n = n_p + n_c + n_p$:
\begin{align} 
    \dot{\zeta}(t) & = A \zeta(t) + B r(t) \label{eq:closed_loop} \\
    u(t) &= E \zeta(t) + F r(t) \label{eq:closed_loop_input}
\end{align}
with $A, \, B, \, E, \, F$ defined in \eqref{eq:ext_matrices} and $\Lambda:=(\textit{I}_{l}-D_p D_c)^{-1}$.

\begin{figure*}[!t]
\begin{align} \label{eq:ext_matrices}
\nonumber A&=\begin{bmatrix}
A_p + B_pD_c\Lambda C_p & B_pC_c + B_pD_c\Lambda D_pC_c & -B_p D_c \Gamma \Gamma^{+} C_p - B_p D_c \Lambda D_p D_c \Gamma \Gamma^{+} C_p\\ B_c \Lambda C_p & A_c + B_c \Lambda D_p C_c & -B_c \Gamma \Gamma^{+} C_p - B_c \Lambda D_p D_c \Gamma \Gamma^{+} C_p\\ \mathbf{0} & \mathbf{0} & A_p-LC_p + L \Gamma \Gamma^{+} C_p
\end{bmatrix}, \,
F=\begin{bmatrix}
 (D_c \Gamma + D_c \Lambda D_p D_c \Gamma) \Gamma^{+}
 \end{bmatrix}, \\
B&=\begin{bmatrix}
 B_pD_c\Gamma\Gamma^{+}+B_pD_c\Lambda D_pD_c\Gamma\Gamma^{+} \\ B_c \Gamma \Gamma^{+} + B_c \Lambda D_pD_c \Gamma \Gamma^{+} \\ - L \Gamma \Gamma^{+}
 \end{bmatrix}, \,\,\,\,\,\,\,\,\,\,\,\,\,\,\,\,\,\,\,\,\,\,\,\,\,\,\,\,\,\,\,\,\,\,\,\,
E=\begin{bmatrix}
 D_c \Lambda C_p & C_c + D_c \Lambda D_p C_c & -(D_c \Gamma + D_c \Lambda D_p D_c \Gamma) \Gamma^{+} C_p
 \end{bmatrix}. 
\end{align}
\end{figure*}

The setup is illustrated in Figure \ref{fig:setup}.

We denote the solution to the closed-loop system with initial condition $\zeta(t_0)$ and sensor corruption $\delta y$ by $\zeta(t;\zeta(t_0),$ $\delta y)$. In the absence of the measurement corruption $\delta y$, i.e., $\delta y(t)=0$ for all $t\geq t_0$, we call system \eqref{eq:closed_loop} our nominal closed-loop system. Under our modelling assumptions (Assumption \ref{assum:model}), the origin is an equilibrium point of our nominal closed-loop system \eqref{eq:closed_loop}.

After having described the system under study with the class of attacks the system might face, we now need a degradation metric to define the safety level under stealthy sensor injection attacks. To this end, we define the notion of safe sets (\cite{Escudero_ECC22}).

\begin{definition}[Safe set $\mathcal{X}_s$] \label{def:safe_set}
The safe set $\mathcal{X}_s \subseteq \mathbb{R}^{n_p}$ for system in \eqref{eq:plant} is the set of states $x_p \in \mathcal{X}_s$ where the safe and proper operation of the system is guaranteed.
\end{definition}

Safe sets exclude, by Definition~\ref{def:safe_set}, critical states of the plant \eqref{eq:plant}. A critical state is a state of the plant that, if reached, compromises the integrity of the system. For instance, the over-pressure in gas pipelines, or the null-distance/negative distance between two vehicles leading to collision are critical states. We can now formally define stealthy sensor injection attacks.

\begin{definition}[Stealthy sensor injection attacks]\label{def:sensor_attack}
Attacks \, that tamper with sensor measurement by injecting signals $\delta y(t)$ to true sensor measurements, $y(t)$, and aim to degrade the operation of the system dynamics in \eqref{eq:plant} by pushing trajectories outside the safe set $\mathcal{X}_s$ while keeping the residuals trajectories inside the stealthy set $\mathcal{E}_r$ for stealthiness.
\end{definition}

Notice that stealthy sensor injection attacks defined in Definition~\ref{def:sensor_attack} only exist if some critical states can be reached while the fault detector does not raise an alarm.

In the absence of measurement corruption $\delta y(t)$, we assume that the controller $\Sigma_c$ in \eqref{eq:control} is designed such that the state of the plant \eqref{eq:plant} evolves within a known safe state set $\mathcal{X}_{s}$, which is compact.
\begin{assumption}[Safe set $\mathcal{X}_{s}$] \label{assum:safe_set}
Given a safe set $\mathcal{X}_{s}\subset \mathbb{R}^{n_p}$ which is compact, the state of the plant \eqref{eq:plant} in the nominal closed-loop system \eqref{eq:closed_loop} (with $\delta y(t)=0$ for all $t \geq t_0$) satisfies $x_p(t,x_p(t_0))\in \mathcal{X}_{s}\subset \mathbb{R}^{n_p}$, for all $t\geq t_0$ and all initial conditions $x_p(t_0)\in\mathbb{R}^{n_{p}}$. Further, the equilibrium of the nominal closed-loop system satisfies $0\in\mathcal{X}_{s}$. \hfill $\Box$
\end{assumption}

\subsection{Safety monitoring} \label{sec:pb}
We propose to detect the presence of stealthy sensor injection attacks $\delta y$ which has not been detected by the fault detector (stealthy), when we know the plant, the controller, the fault-detector models in \eqref{eq:plant}, \eqref{eq:control}, \eqref{eq:filter}, the stealthy set \eqref{eq:stealthy}, and the safe set $\mathcal{X}_s$ defined in Definition~\ref{def:sensor_attack}, while only having access to the input $u(t)$ that might be affected by a corrupted sensor signal $\delta y$. This goal is called \textit{safety monitoring} under stealthy sensor corruption.

The approach proposed in this manuscript is to find an input set $\mathcal{U}$, in which if the control input $u(t)$ does not belong to this set then this means that a stealthy sensor injection attack is driving the plant to critical states. We can now define our research problem as follows.

\begin{problem}
Given the closed-loop system in \eqref{eq:closed_loop}, the stealthy set $\mathcal{E}_r$ in \eqref{eq:stealthy}, and the safe set $\mathcal{X}_s$ defined in Definition~\ref{def:safe_set}, find an input set $\mathcal{U}$ such that if the control input $u(t)$ does not belong to this set $\mathcal{U}$ at some time $t \ge t_0$, i.e., $u(t) \notin \mathcal{U}$, then the system trajectories are not contained in $\mathcal{X}_s$ at some $t \ge t_0$ for all stealthy sensor injection attacks defined in Definition~\ref{def:sensor_attack}.
\end{problem}

\section{Preliminary tool} \label{sec:tools}

\subsection{Stealthy state reachable set}
First, we introduce the definition of the stealthy state reachable set of the closed-loop system \eqref{eq:closed_loop}. This set will be used later on to find an input set $\mathcal{U}$ to solve Problem~1.

\begin{definition} [Stealthy state reachable set]
The stealthy state reachable set $\mathcal{R}_\zeta(t)$ at time $t \in \mathbb{R}_{>0}$ from initial condition $\zeta(t_0) \in \mathbb{R}^n$ is the set of extended state $\zeta(t)$ that satisfy the extended differential equations \eqref{eq:closed_loop}, over all residuals $r(t)$ satisfying $r(t) \in \mathcal{E}_r(\Pi)$ to guarantee the attack stealthiness, i.e.,

\begin{equation}\label{reachset1}
\mathcal{R}_{\zeta}(t):= \left\{ \zeta(t) \left|
\begin{split}
&\zeta(t_0) \in \mathbb{R}^{n},\\
&\zeta(t) \hspace{1mm}\text{satisfies \eqref{eq:closed_loop}},\\
&\text{and }r(t) \in \mathcal{E}_r(\Pi).
\end{split}
\right.
\right\}.
\end{equation}
\end{definition}
Let $\mathcal{R}_\zeta(\infty)$ denote the asymptotic stealthy state reachable set, that is the ultimate bound on $\mathcal{R}_\zeta(t)$, i.e., $\mathcal{R}_\zeta(\infty) := \lim_{t\to\infty} \mathcal{R}_\zeta(t)$.

\begin{remark}
Because $r(t)$ is bounded ($r(t) \in \mathcal{E}_r(\Pi)$), then the asymptotic stealthy state reachable set $\mathcal{R}_\zeta(\infty)$ is compact if $A$ in \eqref{eq:closed_loop} is Hurwitz.
\end{remark}

\subsection{Ellipsoidal bound on $\mathcal{R}_\zeta(\infty)$}
Because the computation of the asymptotic stealthy state reachable set $\mathcal{R}_\zeta(\infty)$ is not tractable, we propose to rely on an outer ellipsoidal approximation $\mathcal{E}_\zeta(Q)$ of $\mathcal{R}_\zeta(\infty)$, i.e., $\mathcal{R}_\zeta(\infty) \subseteq \mathcal{E}_\zeta(Q)$. For the sake of clarity, we will refer to $\mathcal{E}_\zeta(Q)$ as an ellipsoidal bound on $\mathcal{R}_\zeta(\infty)$. To find this ellipsoidal bound, we propose searching for an invariant set for the dynamical system \eqref{eq:closed_loop}, defined as follows.

\begin{definition} [Ellipsoidal bound $\mathcal{E}_\zeta(Q)$]
For $Q \succ 0$, the ellipsoidal set $\mathcal{E}_\zeta(Q)$ is invariant for the dynamical system in \eqref{eq:closed_loop}, if for all initial states $\zeta(t_0) \in \mathcal{E}_\zeta(Q)$, and for all $r(t) \in \mathcal{E}_r(\Pi)$, the trajectories of $\zeta(t)$ in \eqref{eq:closed_loop} satisfy $\zeta(t) \in \mathcal{E}_\zeta(Q), \, \forall t \ge t_0$.
\end{definition}

\cite{Escudero_IET2022} provides sufficient conditions for ellipsoidal sets $\mathcal{E}_\zeta(Q)$ to be invariant for a class of LTI systems as in \eqref{eq:closed_loop} and for some states $\zeta(t)$ that are constrained by a given ellipsoidal set $\mathcal{E}_\varphi(\Phi, \bar{\phi})$, i.e., $\zeta(t) \in \mathcal{E}_\varphi(\Phi, \bar{\varphi})$, with $\Phi \succeq 0$. The method searches for a Lyapunov-like function $V(\zeta) = \zeta^\top Q \zeta$ using Linear Matrix Inequalities (LMIs) (\cite{boyd1994linear}).

We now state the preliminary tool used to find invariant ellipsoidal sets for the closed-loop system in \eqref{eq:closed_loop} with $r(t) \in \mathcal{E}_r(\Pi)$, $\zeta(t) \in \mathcal{E}_\varphi(\Phi, \bar{\varphi})$, $\forall t \ge t_0$.

\begin{lemma}[Invariant Ellipsoidal Set]
Consider the closed-loop system as in \eqref{eq:closed_loop}. If there exist matrix $Q \in \mathbb{R}^{n \times n}$ and constants $\alpha,\,\beta,\,\lambda \in \mathbb{R}_{\geq 0}$ satisfying the following inequalities:
\begin{equation} \label{eq:thun1}
-H - \alpha J - \beta K - \lambda L \succeq 0,
\end{equation}
\begin{equation} \label{eq:thun2}
Q  \succ 0,
\end{equation}
with
\begin{align} \nonumber
H = &\begin{bmatrix}
A^\top Q + Q A\,  & \mathbf{0}\, & Q B\\
\ast\, &  \mathbf{0}\, & \mathbf{0}\\
\ast\, &  \ast\, & \mathbf{0}
\end{bmatrix},\,\,
J = \begin{bmatrix} \label{eq:lem_N}
Q\, & \mathbf{0}\,  & \mathbf{0}\\
\ast\, & -1\, & \mathbf{0}\\
\ast\, & \ast\, & \mathbf{0}
\end{bmatrix},\\
\nonumber
K = &\begin{bmatrix}
\mathbf{0}\, & \mathbf{0}\, & \mathbf{0}\\
\ast\, & 1\, & \mathbf{0}\\
\ast\, & \ast\, & -\Pi
\end{bmatrix},\,\,
L = \begin{bmatrix}
-\Phi\, & \Phi \bar{\varphi}\, & \mathbf{0}\\
\ast\, & 1-\bar{\varphi}^\top \Phi \bar{\varphi}\, & \mathbf{0}\\
\ast\, & \ast\, & \mathbf{0}
\end{bmatrix};
\end{align}
\noindent then,
$\zeta(t_0) ^\top Q \zeta(t_0) \leqslant 1 \Rightarrow
\zeta(t) ^\top Q \zeta(t) \leqslant 1$, for all $t \ge t_0$, $r(t) \in \mathcal{E}_r(\Pi)$, and $\zeta(t) \in \mathcal{E}_\varphi(\Phi, \bar{\varphi})$.
\label{lem:invar}
\end{lemma}

\begin{proof}
Consider first \eqref{eq:thun1}; left and right multiply by $[\zeta(t)^\top, 1, r(t)^\top ]^\top$, and consider $\alpha J$, $\beta K$ and $\lambda L$  as S-procedure terms by positive multipliers $\alpha$, $\beta$ and $\lambda$; this implies:

\begin{equation}
 [\zeta^\top, 1, r^\top] \,H\, [\zeta^\top, 1, r^\top]^\top= \dot{V}(\zeta)  \leqslant 0.
\end{equation}
when
\begin{equation}
 [\zeta^\top, 1, r^\top] \, J \, [\zeta^\top, 1, r^\top]^\top = V(\zeta)-1 \geqslant 0 \Leftrightarrow V(\zeta) \geqslant 1.
\end{equation}
\begin{equation}
 [\zeta^\top, 1, r^\top] \, K \, [\zeta^\top, 1, r^\top]^\top  \geqslant 0 \Leftrightarrow r(t) \in  \mathcal{E}_r(\Pi).
\end{equation}
\begin{equation}
 [\zeta^\top, 1, r^\top] \, L \, [\zeta^\top, 1, r^\top]^\top  \geqslant 0 \Leftrightarrow \zeta(t) \in  \mathcal{E}_\varphi(\Phi, \bar{\varphi}).
\end{equation}
This means that the value of $V(\zeta)$ can only increase under the stated constraints, i.e., $V(\zeta(t_0))\leqslant 1 \Rightarrow  V(\zeta(t))\leqslant 1$ $\forall t \ge t_0$, which concludes the proof.
\end{proof}
\section{Solution to problem 1} \label{sec:result_sec}
In this section, we propose a mathematical framework, built around Lemma~\ref{lem:invar}, to find an input set $\mathcal{U}$ solving Problem~1.

Similarly to the stealthy state reachable set, we define the stealthy input reachable set $\mathcal{R}_u(t)$.

\begin{definition} [Stealthy input reachable set]
The stealthy input reachable set $\mathcal{R}_u(t)$ at time $t \in \mathbb{R}_{>0}$ from initial condition $u(t_0) \in \mathbb{R}^m$ is the set of input $u(t)$ that satisfy the equations \eqref{eq:closed_loop_input}, over all $\zeta(t)$ satisfying the differential equations \eqref{eq:closed_loop} and over all residuals $r(t)$ satisfying $r(t) \in \mathcal{E}_r(\Pi)$ to guarantee the attack's stealthiness, i.e.,
\begin{equation}\label{reachset1}
\mathcal{R}_u(t):= \left\{ u(t) \left|
\begin{split}
&u(t_0) \in \mathbb{R}^{m},\\
&u(t) \hspace{1mm}\text{satisfies \eqref{eq:closed_loop_input}},\\
&\text{and }\zeta(t) \hspace{1mm}\text{satisfies \eqref{eq:closed_loop}},\\
&\text{and }r(t) \in \mathcal{E}_r(\Pi).
\end{split}
\right.
\right\}.
\end{equation}
\end{definition}

Let $\mathcal{R}_u(\infty)$ denote the asymptotic stealthy input reachable set, that is the ultimate bound on $\mathcal{R}_u(t)$, i.e., $\mathcal{R}_u(\infty) := \lim_{t\to\infty} \mathcal{R}_u(t)$.

\begin{remark}
Because $r(t)$ is bounded ($r(t) \in \mathcal{E}_r(\Pi)$) and $\zeta(t)$ is bounded ($\zeta(t) \in \mathcal{R}_\zeta(\infty)$) if condition (ii) in Assumption~1 holds, then the asymptotic stealthy input reachable set $\mathcal{R}_u(\infty)$ exists.
\end{remark}

In this manuscript, we propose to rely on an outer ellipsoidal approximation $\mathcal{E}_u(R)$ of $\mathcal{R}_u(\infty)$, i.e. $\mathcal{R}_u(\infty) \subseteq \mathcal{E}_u(R)$. For the sake of clarity, we will refer $\mathcal{E}_u(R)$ as an ellipsoidal bound on $\mathcal{R}_u(\infty)$, with $R \succ 0$.

Before presenting the proposed approach, we model the safe set $\mathcal{X}_s$ defined in Definition~\ref{def:safe_set} as an ellipsoid $\mathcal{E}_s(\Psi, \bar{\psi})$ written in terms of the extended state $\zeta$ and satisfying:
\begin{align}
(\zeta - \bar{\psi})^\top \Psi (\zeta - \bar{\psi}) \le 1
\end{align}
with
\begin{align}
\Psi = \begin{bmatrix} \Psi_p \, & \mathbf{0}\\
					   \mathbf{0} \, & \mathbf{0}
	   \end{bmatrix}, \,\,\,
\bar{\psi} = \begin{bmatrix} \bar{\psi}_p\\
							\mathbf{0}
			 \end{bmatrix}
\end{align}
for some known positive semi-definite matrix $\Psi_p \in \mathbb{R}^{n_p \times n_p}$ and vector $\bar{\psi}_p \in \mathbb{R}^{n_p}$. Note that $\Psi_p$ is in general rank-deficient as it only constrains some of the plant states $x_p(t)$ --  $\mathcal{E}_s(\Psi, \bar{\psi})$ can even coincide with $\mathbb{R}^{n \times n}$ by picking $\Psi_p = \mathbf{0}$.

\subsection{Approach}
To solve Problem~1 (see Section~\ref{sec:pb}), we propose a two-step procedure: first we compute an ellipsoidal bound $\mathcal{E}_\zeta(Q)$ on the asymptotic stealthy state reachable set of the closed-loop system \eqref{eq:closed_loop} for some states $\zeta(t)$ constrained to remain inside the safe set $\mathcal{E}_s(\Psi, \bar{\psi})$. This ellipsoidal bound describes where the state trajectories $\zeta(t)$ will remain at all times $t \ge t_0$ for some initial conditions $\zeta(t_0) \in \mathcal{E}_\zeta(Q)$ under the presence of stealthy sensor attacks while some parts of the states $\zeta$ are constrained by the safe set $\mathcal{E}_s$. Once this state ellipsoidal bound $\mathcal{E}_\zeta(Q)$ is computed, we compute the input ellipsoidal bound $\mathcal{E}_u(R)$ on $\mathcal{R}_u(\infty)$ when the states $\zeta(t)$ belong to the ellipsoidal bound $\mathcal{E}_\zeta(Q)$ and the residuals $r(t)$ belong to the stealthy set $\mathcal{E}_r(\Pi)$.

If such an input ellipsoidal bound  $\mathcal{E}_u(R)$ can be computed, we can use it to check whether the input signal $u(t)$ is in it. If $u(t)$ does not belong to this set, then it means that the safety of the plant will be violated at some time $t \ge t_0$. Indeed, we only have an outer-approximation of the asymptotic stealthy state reachable set $\mathcal{R}_\zeta(\infty)$ as its exact computation is not tractable. As a result, the corresponding input set $\mathcal{E}_u(R)$ is also an outer-approximation of the asymptotic input reachable set. 

\subsection{Main theorem} \label{sec:result}
The main result of this manuscript is the computation of an input ellipsoidal bound $\mathcal{E}_u(R)$ on $\mathcal{R}_\zeta(\infty)$ of the closed-loop system \eqref{eq:closed_loop} under the presence of stealthy sensor injection attacks $\delta y(t)$ when the state trajectories $\zeta(t)$ are bounded by the state ellipsoidal bound $\mathcal{E}_\zeta(Q)$ at all time $t \ge t_0$, which represents where the state $\zeta(t)$ will remain at all time $t \ge t_0$ when they are constrained by the safe set under the presence of any stealthy sensor injection attacks.

The problem we want to solve is formulated as follows. We want to find an input ellipsoidal bound $\mathcal{E}_u(R)$ such that $u(t)$ satisfy the equation \eqref{eq:closed_loop_input} for all $\zeta(t), \, r(t)$ satisfying $\zeta(t) \in \mathcal{E}_\zeta(Q)$ and $r(t) \in \mathcal{E}_r(\Pi)$ at all time $t \ge t_0$. This means that if the input signal $u(t)$ does not belong to $\mathcal{E}_u(R,\bar{u})$, thus the input signal affected by a stealthy sensor injection attack is violating at some time $t$ at least one constraint, i.e. the state trajectories $\zeta(t)$ are constrained to remain inside the safe set.

\begin{theorem}[Input ellipsoidal bound $\mathcal{E}_u(R)$]
Consider the \linebreak
closed-loop system as in \eqref{eq:closed_loop}, an invariant ellipsoidal set $\mathcal{E}_\zeta(Q)$ with $Q \succ 0$ for the system in \eqref{eq:closed_loop},  and the stealthy set $\mathcal{E}_r(\Pi)$. If there exist matrix $R \in \mathbb{R}^{m \times m}$ and constants $\gamma, \, \tau \in \mathbb{R}_{\geq 0}$ satisfying the following inequalities:
\begin{equation} \label{eq:thun3}
-W - \gamma Y - \tau Z \succeq 0,
\end{equation}
\begin{equation} \label{eq:thun4}
R  \succ 0
\end{equation}
with
\begin{align}
W = &\begin{bmatrix}
E^\top R E^\top\,  & \mathbf{0}\, & E^\top R F\\
\ast\, &  -1\, & \mathbf{0}\\
\ast\, &  \ast\, & F^\top R F
\end{bmatrix},\\
Y = &\begin{bmatrix}
-Q\, & \mathbf{0}\,  & \mathbf{0}\\
\ast\, & 1\, & \mathbf{0}\\
\ast\, & \ast\, & \mathbf{0}
\end{bmatrix},\\
Z = &\begin{bmatrix}
\mathbf{0}\, & \mathbf{0}\,  & \mathbf{0}\\
\ast\, & 1\, & \mathbf{0}\\
\ast\, & \ast\, & -\Pi
\end{bmatrix}
\end{align}
\noindent then,
$u(t_0)^\top R u(t_0) \leqslant 1 \Rightarrow
\zeta(t) ^\top Q \zeta(t) \leqslant 1, r(t)^\top \Pi r(t) \leqslant 1$, for all $t \ge t_0$
\label{th:input}
\end{theorem}

\begin{proof}
Consider \eqref{eq:thun3}; left and right multiply by $[\zeta(t)^\top, 1,$ $r(t)^\top]^\top$, and consider $\gamma Y$ and $\tau Z$ as S-procedure terms by positive multipliers $\gamma$, and $\tau$; this implies with the S-procedure:

\begin{align*}
&-[\zeta^\top, 1, r^\top] \,W\, [\zeta^\top, 1, r^\top]^\top \ge 0 \Leftrightarrow u(t) \in \mathcal{E}_u(R)\\
&\text{when}\\
&[\zeta^\top, 1, r^\top] \, Y \, [\zeta^\top, 1, r^\top]^\top \ge 0 \Leftrightarrow \zeta(t) \in \mathcal{E}_\zeta(Q),\\
&[\zeta^\top, 1, r^\top] \, Z \, [\zeta^\top, 1, r^\top]^\top \ge 0 \Leftrightarrow r(t) \in \mathcal{E}_r(\Pi)
\end{align*}
This means that the input trajectories remain inside the input ellipsoidal bound $\mathcal{E}_u(R)$ for all (i) plant state trajectories inside the invariant ellipsoidal set $\mathcal{E}_\zeta(Q)$ and (ii) residual trajectories $r(t)$ inside the stealthy set $\mathcal{E}_r( \Pi).$
\end{proof}

\subsection{Computation of the input ellipsoidal bound $\mathcal{E}_u(R)$}
To compute the input ellipsoidal bound $\mathcal{E}_u(R)$, first we have to compute the state ellipsoidal bound $\mathcal{E}_\zeta(Q)$ using Lemma~\ref{lem:invar}. Due to the product of $\alpha$ with $J$, the matrix inequality \eqref{eq:thun1} in Lemma~\ref{lem:invar} is not an LMI. To relax it, we will compute the state ellipsoidal bound $\mathcal{E}_\zeta(Q)$ for a fixed $\alpha$. Because we compute an outer-approximation of the stealthy state reachable set $\mathcal{R}_\zeta(\infty)$, we want to find the one with the lowest volume to have the best approximation. Since the volume of an ellipsoid $\mathcal{E}_\varphi(\Phi)$ is proportional to (det($ \Phi$)$^{-1/2}$, we can minimize (log(det($ \Phi$))$^{-1}$ to minimize the volume of the ellipsoid, which is convex and allows to cast the problem as a convex optimization problem (\cite{boyd1994linear}). By solving the convex optimization problem \textbf{OP$_1$} for a fixed $\alpha$  we find a state ellipsoidal bound $\mathcal{E}_\zeta(Q)$ with the lowest volume for the fixed parameter.

\noindent \textbf{OP$_1$: \quad State ellipsoidal bound $\mathcal{E}_\zeta(Q)$}
  \begin{align*}
    \underset{Q, \beta, \lambda}{\text{minimize}} & \qquad \textrm{-log(det(}Q\textrm{))},\\
  \text{subject to}
& \qquad \eqref{eq:thun1}, \eqref{eq:thun2}.
  \end{align*}

After having found the state ellipsoidal bound $\mathcal{E}_\zeta(Q)$, we can now compute the input ellipsoidal bound $\mathcal{E}_u(R)$ using Theorem~\ref{th:input} by solving the convex optimization problem \textbf{OP$_2$}. Similarly, we want to find an input ellipsoidal bound with the lowest volume.

\noindent \textbf{OP$_2$: \quad Input ellipsoidal bound $\mathcal{E}_u(R)$}
  \begin{align*}
    \underset{R, \gamma, \tau}{\text{minimize}} & \qquad \textrm{-log(det(}R\textrm{))},\\
  \text{subject to}
& \qquad \eqref{eq:thun3}, \eqref{eq:thun4}.
  \end{align*}

The procedure to find an input ellipsoidal bound $\mathcal{E}_u(R)$ answering Problem~1 is summarized in Algorithm~1.

\begin{algorithm2e}
  \SetAlgoLined
  \KwResult{Input ellipsoidal bound $\mathcal{E}_u(R)$}
  Init: closed-loop system ($A$, $B$, $E$, $F$), stealthy set ($\mathcal{E}_r(\Pi)$), safe set ($\Psi$, $\bar{\psi}$), attacker's selection matrix ($\Gamma$)\;
  1) Find the state ellipsoidal bound $\mathcal{E}_\zeta (Q)$ by solving \textbf{OP$_1$} for some $\alpha \ge 0$\;
  2) Find the input ellipsoidal bound $\mathcal{E}_u (R)$ by solving \textbf{OP$_2$}\;
  \caption{Compute input ellipsoidal bound $\mathcal{E}_u(R)$}
\end{algorithm2e}

\section{Simulation example} \label{sec:example}
We verify our main result in simulations on a three-tank system. After stating the models of the three-tank system, the controller and the fault-detector, we first describe the attack scenario against the system. Then, we apply our main results using Algorithm~1. Lastly, we show that the proposed approach enables the detection of some stealthy attacks and we highlight the limitations of using reachable sets for detection. We use the solver MOSEK with the YALMIP toolbox on Matlab to solve the optimization problems, and we use Simulink to inject attacks on the system and test the proposed approach.

\subsection{System description}
Consider the three-tank system from \cite{Iqbal_2007} modelled in \eqref{eq:system_example} as an LTI system as in \eqref{eq:plant} with $x_p = [x_{p1}, x_{p2}, x_{p3}]^\top$ ($n_p = 3$) and $u = [u_1, u_2]^\top$ ($m = 2$) where $x_{p}$ is the tank's liquid level [\textrm{cm}], and $u$ is the supply flow rates [\textrm{mL.s$^{-1}$}]. The plant states  $x_{p1}$ and $x_{p2}$ are measured and encompassed in the sensor measurement $y = [y_1, y_2]^T$ ($l = 2$).

\begin{align}\label{eq:system_example}
A_p &= 10^{-4} \times
\begin{bmatrix}
-1.36\, & 0\, & 0.72\\
0\, & -2.29\, & 15.30\\
1.36\, & -1.11\, & -16.02
\end{bmatrix}
,
B_p =
\begin{bmatrix}
64.94\, & 0\\
0\, & 64.94\\
0\, & 0
\end{bmatrix}
\nonumber
\\
C_p &=
\begin{bmatrix}
1\, & 0\, & 0\\
0\, & 1\, & 0
\end{bmatrix}
,
\,
D_p = \begin{bmatrix}
0\, & 0\\
0\, & 0
\end{bmatrix}.
\end{align}
 An observer-based output feedback controller as in \eqref{eq:control} is considered in  closed-loop with controller matrices given in \eqref{eq:controller_example} with $n_c = 3$.
 
 \begin{align}\label{eq:controller_example}
A_c &= 
\begin{bmatrix}
-0.33\, & 0.09\, & -122.34\\
0.10\, & -0.33\, & 114.69\\
-0.01\, & -0.16\, & -0.002
\end{bmatrix}
,
B_c = 10^{-2} \times
\begin{bmatrix}
2.00\, & 0.13\\
0.19\, & 3.80\\
1.49\, & 15.49
\end{bmatrix}
\nonumber
\\
C_c &= 10^{-2} \times
\begin{bmatrix}
-0.47\, & 0.14\, & -188.41\\
0.16\, & -0.45\, & 176.62
\end{bmatrix}
,
\,
D_c = \begin{bmatrix}
0\, & 0\\
0\, & 0
\end{bmatrix}.
\end{align}

A fault detector is implemented as in Fig.~\ref{fig:setup} with gain matrix $L$ given below in \eqref{eq:gain_example}. The fault detector raises an alarm when $r(t) \notin \mathcal{E}_r(\Pi)$ with $\Pi = \text{diag(}1,1\text{)}$.
\begin{equation} \label{eq:gain_example}
L = 10^{-3} \times
\begin{bmatrix}
20.01\, & 1.31\\
1.92\, & 38.02\\
14.93\, & 154.94
\end{bmatrix}
\end{equation}
Notice that the gain matrix $L$ is chosen intentionally here to admit a large class of stealthy attack signals.

The three-tank system operates safely when the states remains inside the safe set $\mathcal{E}_s(\Psi,\bar{\psi})$ defined for $\Psi_p = \text{diag(}10^{-4}, 10, 10^{-4}\text{)}$, and $\bar{\psi_p} = \mathbf{0}$. This means that the plant state $x_{p2}$ is constrained whereas $x_{p1}$ and $x_{p3}$ as the corresponding terms in $\Psi_p$ tend to zero.

\subsection{Attack scenario}
We consider here the worst-case scenario where an adversary is capable of compromising both output measurements $y_1(t)$ and $y_2(t)$ ($s = 2$) by injecting two attack signals $\delta y_1$ and $\delta y_2$ added to the measurements, i.e., $\Gamma = I_2$.

\subsection{Safety monitoring}
Here we want to find an input set $\mathcal{U}$ as a solution to Problem~1 using Algorithm~1. First, we compute the state ellipsoidal bound for $\alpha = 30$. The projection of the state ellipsoidal bound $\mathcal{E}_\zeta(Q)$ onto the $x_p$-hyperplane is drawn on the left-hand side in Figure~\ref{fig:main_result} (green fill) together with the safe set $\mathcal{E}_s$ (blue fill). As we can expect, the plant state $x_{p2}$ is constrained whereas $x_{p1}$ and $x_{p3}$ are not (they can reach almost the entire state space). Then, we compute the input ellipsoidal bound $\mathcal{E}_u(R)$ where the resulting $R$ is given in \eqref{eq:R_result}. The input ellipsoidal bound is drawn and zoomed in on the right-hand side in Figure~\ref{fig:main_result} (red fill).
\begin{equation} \label{eq:R_result}
R = 10^5 \times
\begin{bmatrix} 3.99\, & 4.26\\
			     4.26\, & 4.55
			     \end{bmatrix}
\end{equation}

\begin{figure*}[h]
    \centering
    \begin{subfigure}[b]{0.9\textwidth} 
    	\includegraphics[width=0.5\textwidth]{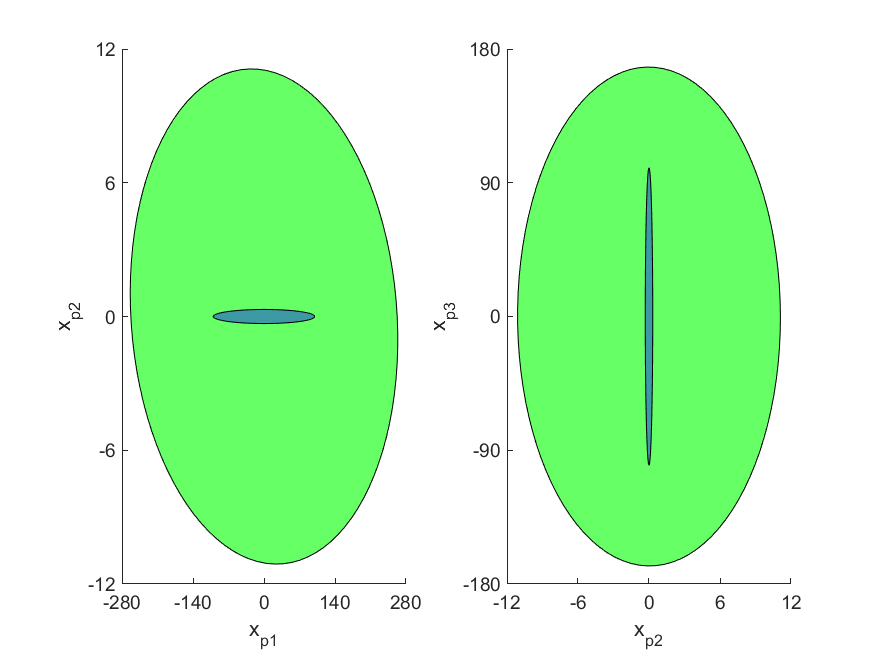}
    	\includegraphics[width=0.5\textwidth]{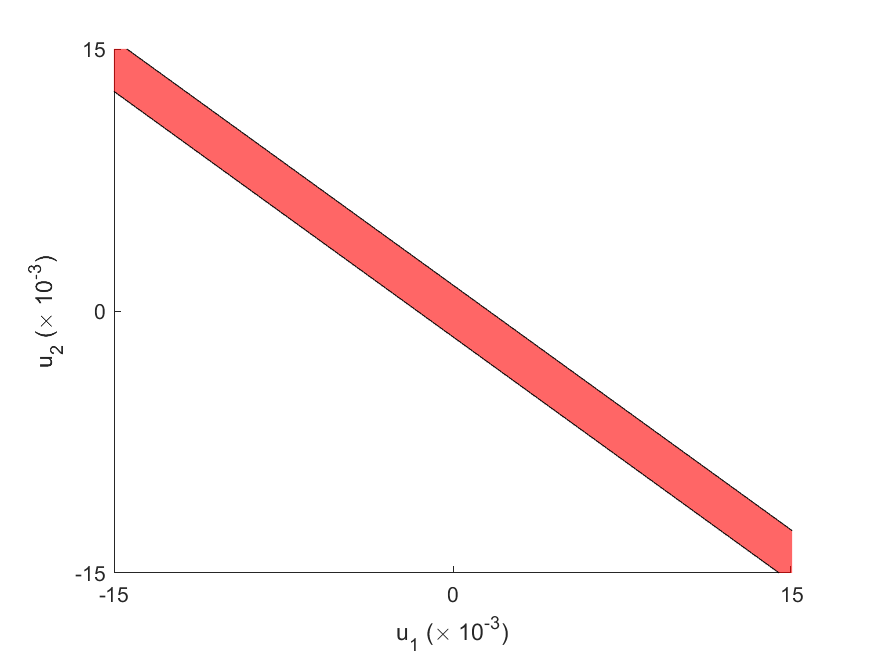}
    \end{subfigure}
    \caption{(Left-hand side) Projection of the state ellipsoidal bound $\mathcal{E}_\zeta(Q)$ onto the $x_p$-hyperplane (green fill), safe set $\mathcal{E}_s(\Psi, \bar{\psi})$ (blue fill) - (Right-hand side) Input ellipsoidal bound $\mathcal{E}_u(R)$ zoomed in (red fill)}\label{fig:main_result}
\end{figure*}

After having found the input ellipsoidal bound $\mathcal{E}_u(R)$, we test the proposed detection scheme based on the input ellipsoidal bound to verify if we can detect stealthy sensor injection attacks by checking whether the control input $u(t)$ belongs to the input ellipsoidal bound or not. Figure~\ref{fig:sim_result} shows the simulation result with, from the top to the bottom on the left-hand side: the plant states $x_p(t)$, a safety check (takes $1$ if the plant states do not belong to the safe set; otherwise $0$), the control input $u(t)$, and the detection result (returns $1$ if the control input do not belong to the input ellipsoidal bound $\mathcal{E}_u(R)$; otherwise $0$). In Figure~\ref{fig:sim_result}, from the top to the bottom on the right-hand side, the attack signals $\delta y(t)$, the residuals $r(t)$, and a stealthiness check (returns $1$ if the residuals belong to the stealthy set; otherwise $0$) are drawn. Three attacks have been launched on the closed-loop system : first attack starts at 10 \textrm{s} and ends at 110 \textrm{s}, second attack starts at 175 \textrm{s} and ends at 275 \textrm{s}, and last attack starts at 340 \textrm{s} and ends at 440 \textrm{s}. As we can observe, the attack signals $\delta y$ are carefully choosen to keep the fault detector below its detection threshold, i.e. the residualts $r(t) \in \mathcal{E}_s(\Pi)$ (see the stealthiness check in Figure~\ref{fig:sim_result}). Each attack violates the safety of the plant at some time $t$. For each attack, the safety is first violated respectively at 42.3 \textrm{s}, 188.1 \textrm{s}, and 388.5 \textrm{s}. As we can pinpoint, the attack signals $\delta y(t)$ affect the control input $u(t)$ to compromise the integrity of the plant. For each attack, the detection occurs at 59.8 \textrm{s}, 218.4 \textrm{s}, and no detection is triggered for the last attack.

First, we can see that the proposed approach succeeded in revealing the presence of stealthy sensor injection attacks by monitoring only the control input $u(t)$ for the first two attacks. However, it does not detect the presence of the last attack. For this last attack, we can observe that the control input $u(t)$ and the plant state $x_p(t)$ are very slightly affected by the attack signals which allow the attack to escape the detection. Second, the time before detection that is the delay of detection after the occurrence of the safety violation can be studied. Here, the detection occurs after 16.5 \textrm{s} for the first attack and 30.3 \textrm{s} for the second attack. Both statements are likely because the proposed approach is based on ellipsoidal bounds of the reachable sets, meaning that there exists an approximation error. So, at the time before the safety violation, the control input affected by the attacks are most likely inside  $\mathcal{R}_u(\infty)$, and thus inside the input ellipsoidal bound $\mathcal{E}_u(R)$. Before the safety violation, the control input is outside $\mathcal{R}_u(\infty)$, but inside $\mathcal{E}_u(R)$. This means that the control input will drive the plant outside the safe set but it cannot be detected by the proposed approach. For the last attack, the control input is most likely operating in this zone, i.e. $u(t) \notin \mathcal{R}_u(\infty)$ and $u(t) \in \mathcal{E}_u(R)$. However, for the first two attacks, the control input at the time detection crosses the boundary of the input ellipsoidal bound leading to the detection. This limitation could be overcome first by using other sets than ellipsoidal sets that often lead to a large approximation error, and second by considering the time evolution of the reachable set instead of using the asymptotic reachable set.

\begin{figure*}[h]
    \centering
    \begin{subfigure}[b]{\textwidth} 
  	 \includegraphics[width=0.5\textwidth, trim = 0.9cm 0cm 1.1cm 0cm, clip]{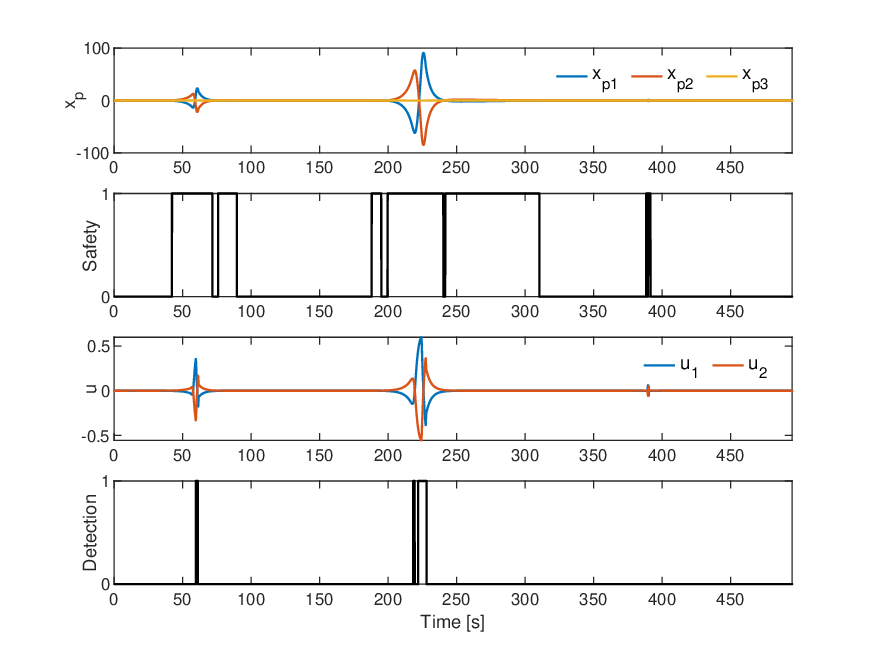}
        \includegraphics[width=0.5\textwidth, trim = 1cm 0cm 1.1cm 0cm, clip]{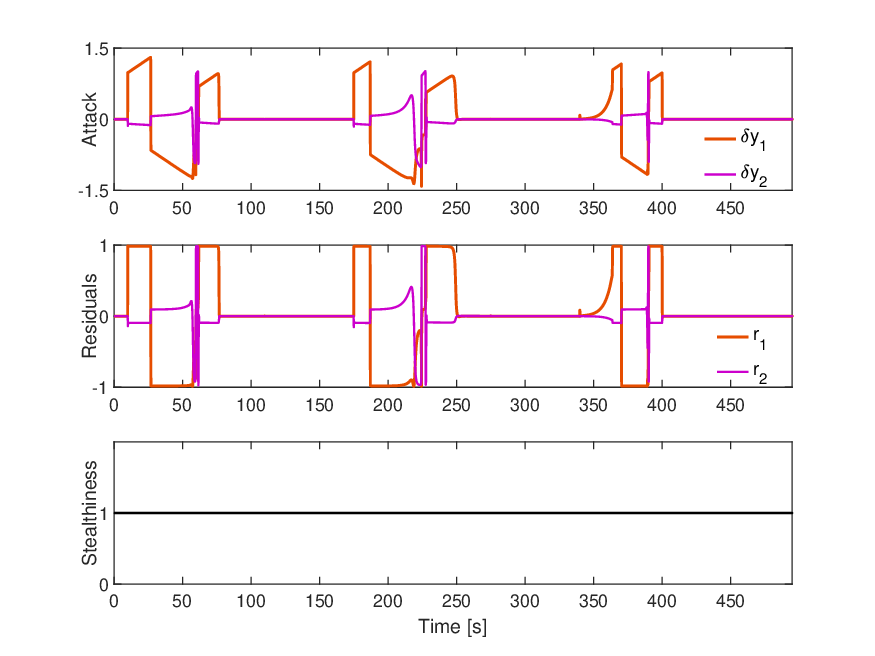}
    \end{subfigure}

    \caption{Stealthy sensor injection attacks $\delta y$ on the closed-loop system: plant states $x_p(t)$ and safety check (top left corner), control input $u(t)$ and detection result (bottom left corner), attack signals $\delta y(t)$ (top right corner), residuals $r(t)$ and stealthiness check (bottom right corner)}\label{fig:sim_result}
\end{figure*}

\section{Conclusion}
We have proposed a set-theoretic method to design a detection scheme to reveal the presence of stealthy sensor injection attacks by only monitoring the control input to the plant. The novelty of this work lies in monitoring the control input to reveal attacks evading traditional fault-detectors while pushing the plant into critical states. Some limitations have been highlighted that will be the subject of future research works. First, we will explore the possibility of computing an input ellipsoidal bound that will evolve with time to cope with the main limitation of the approach. Second, sets other than ellipsoidal sets will be considered to reduce the error of the outer-approximation of the reachable sets.

\bibliography{michellebib}

\begin{thebibliography}{19}
\providecommand{\natexlab}[1]{#1}
\providecommand{\url}[1]{\texttt{#1}}
\providecommand{\urlprefix}{URL }
\expandafter\ifx\csname urlstyle\endcsname\relax
  \providecommand{\doi}[1]{doi:\discretionary{}{}{}#1}\else
  \providecommand{\doi}{doi:\discretionary{}{}{}\begingroup
  \urlstyle{rm}\Url}\fi

\bibitem[{Azzam et~al.(2022)Azzam, Pasquale, Provan, and Nuseibeh}]{Azzam_2022}
Azzam, M., Pasquale, L., Provan, G., and Nuseibeh, B. (2022).
\newblock Efficient predictive monitoring of linear time-invariant systems
  under stealthy attacks.
\newblock \emph{IEEE Transactions on Control Systems Technology}, 1--13.

\bibitem[{Boyd et~al.(1994)Boyd, El~Ghaoui, Feron, and
  Balakrishnan}]{boyd1994linear}
Boyd, S., El~Ghaoui, L., Feron, E., and Balakrishnan, V. (1994).
\newblock \emph{Linear matrix inequalities in system and control theory},
  volume~15.
\newblock SIAM.

\bibitem[{Cardenas et~al.(2009)Cardenas, Amin, Sinopoli, Perrig, and
  Sastry}]{Cardenas_2009}
Cardenas, A.A., Amin, S., Sinopoli, B., Perrig, A., and Sastry, S. (2009).
\newblock Challenges for securing cyber physical systems.
\newblock \emph{Proc. 1St Workshop Cyber-Phys. Syst. Security DHS}.

\bibitem[{Chong et~al.(2019)Chong, Sandberg, and Teixeira}]{Chong_2019}
Chong, M.S., Sandberg, H., and Teixeira, A.M. (2019).
\newblock A tutorial introduction to security and privacy for cyber-physical
  systems.
\newblock In \emph{2019 18th European Control Conference (ECC)}, 968--978.

\bibitem[{Escudero et~al.(2022{\natexlab{a}})Escudero, Massioni, Zamaï, and
  Raison}]{Escudero_IET2022}
Escudero, C., Massioni, P., Zamaï, E., and Raison, B. (2022{\natexlab{a}}).
\newblock Analysis, prevention, and feasibility assessment of stealthy ageing
  attacks on dynamical systems.
\newblock \emph{IET Control Theory \& Applications}, 16(4), 381--397.

\bibitem[{Escudero et~al.(2022{\natexlab{b}})Escudero, Murguia, Massioni, and
  Zamaï}]{Escudero_ECC22}
Escudero, C., Murguia, C., Massioni, P., and Zamaï, E. (2022{\natexlab{b}}).
\newblock Enforcing safety under actuator injection attacks through input
  filtering.
\newblock In \emph{2022 European Control Conference (ECC)}, 1521--1528.

\bibitem[{Firoozjaei et~al.(2022)Firoozjaei, Mahmoudyar, Baseri, and
  Ghorbani}]{Firoozjaei2022}
Firoozjaei, M.D., Mahmoudyar, N., Baseri, Y., and Ghorbani, A.A. (2022).
\newblock An evaluation framework for industrial control system cyber
  incidents.
\newblock \emph{International Journal of Critical Infrastructure Protection},
  36, 100487.

\bibitem[{Hu et~al.(2019)Hu, Li, Yang, Sun, Sun, and Wang}]{Hu_2019}
Hu, Y., Li, H., Yang, H., Sun, Y., Sun, L., and Wang, Z. (2019).
\newblock Detecting stealthy attacks against industrial control systems based
  on residual skewness analysis.
\newblock \emph{Journal on Wireless Communications and Networking}.

\bibitem[{Iqbal et~al.(2007)Iqbal, Butt, and Bhatti}]{Iqbal_2007}
Iqbal, M., Butt, Q.R., and Bhatti, A.I. (2007).
\newblock Linear model based diagnostic framework of three tank system.
\newblock In \emph{11th WSEAS International Conference on Systems}.

\bibitem[{Kurzhanski and Varaiya(2000)}]{kurzhanski2000ellipsoidal}
Kurzhanski, A.B. and Varaiya, P. (2000).
\newblock Ellipsoidal techniques for reachability analysis: internal
  approximation.
\newblock \emph{Systems \& control letters}, 41(3), 201--211.

\bibitem[{Lee(2008)}]{Lee_CPS}
Lee, E.A. (2008).
\newblock Cyber physical systems: Design challenges.
\newblock In \emph{2008 11th IEEE International Symposium on Object and
  Component-Oriented Real-Time Distributed Computing (ISORC)}, 363--369.

\bibitem[{Li et~al.(2022)Li, Wang, Shen, and Xie}]{Li_2022b}
Li, J., Wang, Z., Shen, Y., and Xie, L. (2022).
\newblock Security synthesis for cyber-physical systems.
\newblock \emph{IEEE Transactions on Systems, Man, and Cybernetics: Systems},
  1--11.

\bibitem[{Liu et~al.(2021)Liu, Niu, and Qin}]{Liu_2021}
Liu, H., Niu, B., and Qin, J. (2021).
\newblock Reachability analysis for linear discrete-time systems under stealthy
  cyber attacks.
\newblock \emph{IEEE Transactions on Automatic Control}, 66(9), 4444--4451.

\bibitem[{Milošević et~al.(2020)Milošević, Sandberg, and
  Johansson}]{Milosevic_2020}
Milošević, J., Sandberg, H., and Johansson, K.H. (2020).
\newblock Estimating the impact of cyber-attack strategies for stochastic
  networked control systems.
\newblock \emph{IEEE Transactions on Control of Network Systems}, 7(2),
  747--757.

\bibitem[{Murguia et~al.(2020)Murguia, Shames, Ruths, and
  Nešić}]{MurguiaAutomatica}
Murguia, C., Shames, I., Ruths, J., and Nešić, D. (2020).
\newblock Security metrics and synthesis of secure control systems.
\newblock \emph{Automatica}, 115, 108757.

\bibitem[{Sandberg et~al.(2022)Sandberg, Gupta, and
  Johansson}]{annual_rev_Sandberg_2021}
Sandberg, H., Gupta, V., and Johansson, K.H. (2022).
\newblock Secure networked control systems.
\newblock \emph{Annual Review of Control, Robotics, and Autonomous Systems},
  5(1), null.

\bibitem[{Teixeira et~al.(2012)Teixeira, P{\'e}rez, Sandberg, and
  Johansson}]{Teixeira}
Teixeira, A., P{\'e}rez, D., Sandberg, H., and Johansson, K.H. (2012).
\newblock Attack models and scenarios for networked control systems.
\newblock In \emph{Proceedings of the 1st International Conference on High
  Confidence Networked Systems}, HiCoNS '12, 55--64. ACM, New York, NY, USA.

\bibitem[{Yang et~al.(2021)Yang, Murguia, Kuijper, and Nešić}]{Yang_2021}
Yang, T., Murguia, C., Kuijper, M., and Nešić, D. (2021).
\newblock An unknown input multiobserver approach for estimation and control
  under adversarial attacks.
\newblock \emph{IEEE Transactions on Control of Network Systems}, 8(1),
  475--486.

\bibitem[{Zhang et~al.(2022)Zhang, Liu, Pang, Xia, and Liu}]{Zhang_2022}
Zhang, Q., Liu, K., Pang, Z., Xia, Y., and Liu, T. (2022).
\newblock Reachability analysis of cyber-physical systems under stealthy
  attacks.
\newblock \emph{IEEE Transactions on Cybernetics}, 52(6), 4926--4934.

\end{thebibliography}

\end{document}